\documentclass[11pt]{amsart}
\usepackage{epsfig}
\usepackage{amsbsy,latexsym}
\usepackage{amsmath}
\usepackage{amssymb, mathrsfs}
\usepackage[mathscr]{eucal}
\usepackage{hyperref}
\usepackage{amsfonts}
\usepackage{amsthm}
\usepackage{amsmath}
\usepackage{amsfonts}
\usepackage{latexsym}
\usepackage{amssymb}
\usepackage{amscd}
\usepackage[latin1]{inputenc}
\usepackage{verbatim}

\usepackage[english]{babel}

\newtheorem{definition}{Definition}[section]
\newtheorem{theorem}[definition]{Theorem}
\newtheorem{lemma}[definition]{Lemma}
\newtheorem{corollary}[definition]{Corollary}
\newtheorem{proposition}[definition]{Proposition}
\theoremstyle{definition}
\newtheorem{remark}[definition]{Remark}

\newtheorem{example}[definition]{Example}

\newcommand{\C}{\mathbb{C}}

\newcommand\tr{ \operatorname{Tr} }
\newcommand\spn{ \operatorname{span} }

\def\ket#1{| #1 \rangle}
\def\bra#1{\langle #1 |}
\def\kb#1#2{|#1\rangle\!\langle #2 |}

\title[Quantum Complementarity and Operator Structures]{Quantum Complementarity and Operator Structures}

\begin{document}

\author[D.~W. Kribs, J.~Levick, M.~I.~Nelson, R. Pereira, M.~Rahaman]{David~W.~Kribs$^{1,2}$, Jeremy Levick$^{2,3}$, Mike~I.~Nelson$^{1}$, Rajesh Pereira$^{1}$, Mizanur Rahaman$^{4}$}

\address{$^1$Department of Mathematics \& Statistics, University of Guelph, Guelph, ON, Canada N1G 2W1}
\address{$^2$Institute for Quantum Computing, University of Waterloo, Waterloo, ON, Canada N2L 3G1}
\address{$^2$Department of Mathematics \& Applied Mathematics, University of Cape Town, Cape Town 7700, South Africa}
\address{$^4$Department of Pure Mathematics, University of Waterloo, Waterloo, ON, Canada N2L 3G1}



\maketitle


\begin{abstract}
We establish operator structure identities for quantum channels and their error-correcting and private codes, emphasizing the complementarity relationship between the two perspectives. Relevant structures include correctable and private operator algebras, and operator spaces such as multiplicative domains and nullspaces of quantum channels and their complementary maps. For the case of privatizing to quantum states, we also derive dimension inequalities on the associated operator algebras that further quantify the trade-off between correction and privacy. 
\end{abstract}

\section{Introduction}

The complementary relationship between quantum error correction and quantum privacy is well established. Most relevant to the present work, a quantum code is correctable for a quantum channel if and only if it is private for the channel's complementary map \cite{kks}. This linkage of two fundamental topics in quantum information has more recently \cite{cklt} been extended to the complementarity of appropriate notions of correctable operator algebras \cite{bkk1,bkk2} and private subsystems and algebras \cite{ambainis,boykin,bartlett2,bartlett1,church,jochym1,jochym,lkp,ljklp,cklt}, and to a setting that embraces descriptions of hybrid classical and quantum information \cite{kup,devetak2005capacity,hsieh2010entanglement,yard2005simultaneous,kremsky2008classical,braunstein2011zero,grassl2017codes,klappen2018}.
Quantum error correction as a subject is considerably more developed than the theory of private quantum codes and algebras, with origins going back over two decades to the beginnings of modern quantum information science~\cite{shor1995pw,steane1996error,gottesman1996d,bennett1996ch,knill-lafl,knill2000theory}. The complementarity relationship suggests that developments in one field could at the least influence progress in the other. Of particular interest here, we note how completely positive map multiplicative domain structures  and techniques \cite{choi1} have been used to describe traditional quantum error correcting (subspace and subsystem) codes in terms of operator structures associated with quantum channels \cite{c-j-k,johnston,miza}.

Our goals with this paper are thus threefold. We first extend the multiplicative domain description of a quantum channel's error correcting codes to the setting of (finite-dimensional) correctable algebras. We next identify appropriate operator structures, determined by certain operator null spaces, that describe a channel's private codes and algebras and we show explicitly how they are related to the corresponding multiplicative domains, giving an operator structure depiction of quantum complementarity. Finally, we push our analysis further in the distinguished special case of algebras privatized to quantum states and derive a number of dimension inequalities relating correctable and private pairs of algebras, further quantifying the trade-off between correction and privacy in that case. 

This paper is organized as follows. The next section contains requisite background material on complementary quantum channels, correctable algebras, and private algebras, and includes a new simple proof of complementarity in the ideal case. In Section~3, we extend the multiplicative domain descriptions of quantum error correction to the case of general (finite-dimensional) algebras, identify appropriate null space structures that characterize private algebras, and explicitly relate the complementary structures to each other. We also consider the case of unital channels and uncover extra features for that subclass. In Section~4, we extend the analysis of these identities and structures to derive the aforementioned dimension inequalities that relate the complementary sizes of correctable and private algebra pairs for a given channel. We conclude with a brief discussion of connections and potential future directions.

\section{Complementary Channels and Correctable vs Private Algebras}

We use standard quantum information notation throughout the paper as in~\cite{nc00}. In this section we introduce our requisite preliminary notions: complementary quantum channels, and then correctable and private operator algebras based on the formulation from \cite{cklt}. We shall work with finite-dimensional Hilbert spaces $\mathcal H$, where the sets of linear, trace class, and bounded operators coincide: $\mathcal L(\mathcal H) = \mathcal T(\mathcal H) = \mathcal B(\mathcal H)$, and so for ease of presentation we use $\mathcal L(\mathcal H)$ to denote these sets.

\subsection{Complementary Quantum Channels} By a {\it channel}, we mean a completely positive and trace preserving map $\Phi : \mathcal L(\mathcal H_A) \rightarrow \mathcal L(\mathcal H_A)$ on a Hilbert space $\mathcal H_A$. The Stinespring dilation theorem \cite{stinespring} gives a Hilbert space $\mathcal H_C$ (with $\dim \mathcal H_C \leq (\dim\mathcal H_A)^2$), a state $\ket{\psi_C}\in\mathcal H_C$ and a unitary $U$ on $\mathcal H_A \otimes \mathcal H_C$ such that for all $\rho\in\mathcal L(\mathcal H_A)$,
\begin{equation}\label{stinespring}
\Phi(\rho) = \tr_C \circ \, \mathcal U \, (\rho \otimes \kb{\psi_C}{\psi_C}) = \tr_C \circ \, \mathcal V \, (\rho),
\end{equation}
where here $\tr_C$ denotes the partial trace map from $\mathcal L (\mathcal H_A \otimes \mathcal H_C)$ to $\mathcal L(\mathcal H_A)$, the map $\mathcal U(\cdot) = U(\cdot) U^*$, and $\mathcal V(\cdot) = V(\cdot) V^*$ is the map implemented by the isometry $V: \mathcal H_A \rightarrow \mathcal H_A \otimes \mathcal H_C$ defined by $V\ket{\psi} = U(\ket{\psi}\otimes \ket{\psi_C})$.

Given such a channel $\Phi$, we shall work with the {\it complementary map} $\Phi^C$ defined from $\mathcal L(\mathcal H_A)$ to $\mathcal L(\mathcal H_C)$ as follows:
\begin{equation}\label{compchannel}
\Phi^C(\rho) = \tr_A \circ \, \mathcal V \, (\rho).
\end{equation}
The channel $\Phi^C$ is unique in the following sense: Given any channel $\Phi' : \mathcal L(\mathcal H_A)\rightarrow \mathcal L(\mathcal H_{C'})$ such that $\Phi(\rho) = \tr_{C'} V_1 \rho V_1^*$ and $\Phi'(\rho) = \tr_A V_1 \rho V_1^*$ with $V_1 : \mathcal H_A \rightarrow \mathcal H_A \otimes \mathcal H_{C'}$, there is a partial isometry $W: \mathcal H_C \rightarrow \mathcal H_{C'}$ such that $\Phi'(\cdot) = W \Phi^C (\cdot) W^*$. When $\mathcal H_C$ and $\mathcal H_{C'}$ are minimal (which occurs when both have dimension equal to the so-called Choi rank of $\Phi$), $W$ is a unitary operator. For brevity, unless otherwise noted we shall assume $\Phi^C$ is the (unique) minimal complementary map in this sense.


See \cite{holevo2,holevo,entng-brkng} for further details on complementary channels. Here we only note additionally how the Kraus operators for the two maps $\Phi$, $\Phi^C$ are related: If $\Phi$ has operator-sum representation $\Phi(\rho) = \sum_i V_i \rho V_i^*$ with Kraus operators $V_i\in\mathcal L(\mathcal H_A)$ (which are guaranteed to exist by the Stinespring theorem), then the complementary map has representation,
\[
\Phi^C(\rho) = \sum_{i,j} \tr(\rho V_j^* V_i) \, \kb{i}{j},
\]
where $\{ \ket{i} \}$ is a canonical basis for $\mathbb C^d$ identified with $\mathcal H_C$, and $d = \dim \mathcal H_C$.

\subsection{Correctable Algebras} The general framework for error correction, called ``operator algebra quantum error correction'' (OAQEC) \cite{bkk1,bkk2,bkk3}, when applied to the finite-dimensional case makes use of the structure theory for finite-dimensional von Neumann algebras (or equivalently, C$^*$-algebras). Specifically, codes are identified with algebras that up to unitarily equivalence can be decomposed as $\mathcal A = \oplus_k (I_{m_k}\otimes M_{n_k})$, where $M_n$ is the set of $n \times n$ complex matrices. An OAQEC code is described as follows in each of the Schr\"{o}dinger and Heisenberg pictures for  quantum dynamics. We shall use the notation $\Phi^\dagger$ for the dual map of $\Phi$ defined via the trace inner product: $\tr (\Phi(\rho) X) = \tr (\rho \, \Phi^\dagger(X))$.

\begin{definition}
Let $\mathcal H$ be a (finite-dimensional) Hilbert space and let $Q$ be a projection on $\mathcal H$. Given a channel $\Phi : \mathcal L(\mathcal H) \rightarrow \mathcal L(\mathcal H)$, a von Neumann subalgebra $\mathcal A \subseteq \mathcal L(Q \mathcal H)$ is {\bf  correctable for $\Phi$ with respect to $Q$} if there exists a channel $\mathcal R : \mathcal L (\mathcal H_C) \rightarrow \mathcal A$ such that
\begin{equation}\label{correctableHeisen}
\mathcal P_Q \circ \Phi^\dagger \circ \mathcal R^\dagger = \mbox{id}_{\mathcal A},
\end{equation}
where $\mathcal P_{Q}$ is the compression map $\mathcal P_Q(\cdot) = Q (\cdot ) Q$. When $Q = I$ we simply say $\mathcal A$ is {\bf correctable for $\Phi$}.
\end{definition}

The case of standard (Knill-Laflamme) error correction is captured with algebras $\mathcal A = P_{\mathcal C} \mathcal L(\mathcal H) P_{\mathcal C}$, where $\mathcal C$ is a subspace of $\mathcal H$ and $Q = P_{\mathcal C}$. When $\mathcal C = \mathcal H_A \otimes \mathcal H_B$ has some tensor decomposition, correctable algebras $\mathcal A = P_{\mathcal C} (I_A \otimes \mathcal L(\mathcal H_B)) P_{\mathcal C}$ are ``operator subsystem codes'' \cite{klpo,klp} when $\dim \mathcal H_B > 1$ and classical codes when $\dim \mathcal H_B =1$.  Algebras $\mathcal A$ comprised of direct sums give mixtures of these various possibilities and allow for hybrid classical and quantum information encodings \cite{bkk1,bkk2,kup}. Such an algebra, with direct sum decomposition as above, is correctable for $\Phi$ with respect to its unit projection if and only if for all density operators $\sigma_k^{(i)}$ and probability distributions $p_k$,  there is a channel $\mathcal R$ on $\mathcal H$ and density operators $\sigma_k^{(1)'}$ such that
\begin{equation}\label{correctableSchro}
(\mathcal R \circ \Phi)\, \left( \sum_k p_k  (\sigma_k^{(1)}\otimes \sigma_k^{(2)})    \right)  =    \sum_k p_k  (\sigma_k^{(1)'}\otimes \sigma_k^{(2)}) .
\end{equation}

\subsection{Private Algebras} Private quantum channels were initially introduced as the quantum analogue of the classical one-time pad \cite{ambainis,boykin}. Over subsequent years the idea has been distilled and extended, culminating in the following general notion of what are called ``private algebras'' (see \cite{cklt} and references therein), a notion most cleanly presented in the Heisenberg picture.

\begin{definition}
Let $\mathcal H$ be a (finite-dimensional) Hilbert space and let $Q$ be a projection on $\mathcal H$. Given a channel $\Phi : \mathcal L(\mathcal H) \rightarrow \mathcal L(\mathcal H)$, a von Neumann subalgebra $\mathcal A \subseteq \mathcal L(Q \mathcal H)$ is {\bf private for $\Phi$ with respect to $Q$} if
\begin{equation}\label{privateHeisen}
\mathcal P_Q \circ \Phi^\dagger (\mathcal L(\mathcal H)) \subseteq \mathcal A^\prime = \{ X\in \mathcal L(Q \mathcal H) \,\, | \, \, [X,A]=0 \,\, \forall A \in \mathcal A \}.
\end{equation}
When $Q = I$ we simply say $\mathcal A$ is {\bf private for $\Phi$}.
\end{definition}

This definition is motivated by the notion of an ``operator private subsystem''  \cite{brs,ljklp}: Suppose we have $ H = (\mathcal H_A \otimes \mathcal H_B) \oplus (\mathcal H_A \otimes \mathcal H_B)^\perp$ and a channel $\Phi$ on $\mathcal H$. Then the subsystem $B$ is called an operator private subsystem for $\Phi$ if $\Phi \circ \mathcal P_{\mathcal C} = (\Psi \otimes \tr)\circ \mathcal P_{\mathcal C}$ for some channel $\Psi : \mathcal L(\mathcal H_A)\rightarrow \mathcal L(\mathcal H)$, where $\mathcal P_{\mathcal C}(\cdot) = P_{\mathcal C} (\cdot) P_{\mathcal C}$ with $P_{\mathcal C}$ the projection of $\mathcal H$ onto $C = \mathcal H_A \otimes \mathcal H_B$. One can check through direct calculation and application of the dual map relation that this is equivalent to: $\mathcal P_C \circ \Phi^\dagger  (\mathcal L(\mathcal H)) \subseteq \mathcal L(\mathcal H_A)\otimes I_B = (I_A \otimes \mathcal L(\mathcal H_B))^\prime$; in other words, that the algebra $\mathcal A = I_A \otimes \mathcal L(\mathcal H_B)$ is private for $\Phi$ with respect to $P_{\mathcal C}$.

As articulated in \cite{cklt}, use of the ``private'' terminology is motivated by the fact that any information stored in the operator private subsystem $B$ completely decoheres under the action of $\Phi$. From the Heisenberg perspective, observables on the output system evolve under $\Phi$ to observables having the same measurement statistics with respect to the subsystem $B$. 
For more general private subalgebras though, not all information about observables in the algebra $\mathcal A$ is lost under the action of $\Phi$, just the quantum information: more precisely, the only obtainable information about $\mathcal A$ after an application of the channel is the classical information contained in its centre $\mathcal Z(\mathcal A)=\mathcal A \cap \mathcal A^\prime$. We recover the original notion of privacy when $\mathcal A$ is a von Neumann algebra factor ($\mathcal Z(\mathcal A)=\mathbb{C}I$), and factors of type~I specifically correspond to operator private subsystems. The above definition allows for more general private scenarios as depicted by more general algebras.

\subsection{Complementarity for Perfect Correction and Privacy} We conclude this section by presenting a simple new proof of  complementarity between quantum error correction and privacy in the ideal ($\varepsilon = 0$) case of perfect correction and privacy. We first recall the testable conditions for correctable algebras derived in \cite{bkk1,bkk2}, which in turn built upon the central Knill-Laflamme conditions for standard \cite{knill-lafl} and operator \cite{klp,klpo} quantum error correction.

\begin{theorem}\label{testable}
Let $\mathcal H$ be a Hilbert space and let $Q$ be a projection on $\mathcal H$. Given a channel $\Phi : \mathcal L(\mathcal H) \rightarrow \mathcal L(\mathcal H)$, an algebra $\mathcal A \subseteq \mathcal L(Q \mathcal H)$ is correctable for $\Phi$ with respect to $Q$ if and only if
\begin{equation}\label{testablecond}
[Q V_j^* V_i Q, X]=0  \quad \forall X\in \mathcal A, \, \, \forall i,j . 
\end{equation}
\end{theorem}

The approximate version of the following theorem was established via dilation theory techniques separately in the finite (\cite{kks}) and infinite (\cite{cklt}) dimensional cases. Our proof below for the ideal case is different in that it makes use of Kraus operator representations and the relevant operator structures.

\begin{proposition}\label{perfectcase}
Let $\mathcal A$ be a subalgebra of $\mathcal L(Q \mathcal H)$, for some Hilbert space $\mathcal H$ and projection $Q$. Let $\Phi$ be a channel on $\mathcal H$ with complementary channel $\Phi^C$. Then $\mathcal{A}$ is correctable for $\Phi$ with respect to $Q$ if and only if $\mathcal A$ is private for $\Phi^C$ with respect to $Q$.
\end{proposition}

\begin{proof}
Suppose first that $\mathcal A$ is correctable for $\Phi$ with respect to $Q$. Then Eqs.~(\ref{testablecond}) hold. So we let $\rho\in\mathcal L(\mathcal H)$, $Y\in \mathcal L(\mathcal H_C)$, and $X\in \mathcal A$, and compute two identities as follows:
\begin{eqnarray*}
\tr(Q (\Phi^C)^\dagger(Y) Q X \rho) &=& \tr(Y \Phi^C(QX \rho Q)) \\
&=&  \tr \Big(Y \big( \sum_{i,j} \tr( QX\rho Q V_j^* V_i) \big) \kb{i}{j} \Big) \\
&=&  \sum_{i,j} \tr (\rho Q V_j^* V_i Q X)  \tr (Y\kb{i}{j}),
\end{eqnarray*}
and,
\begin{eqnarray*}
\tr( X Q (\Phi^C)^\dagger(Y) Q \rho) &=& \tr(X Q (\Phi^C)^\dagger (Y) Q \rho ) \\
&=& \tr((\Phi^C)^\dagger (Y) Q \rho X Q) \\
&=& \tr(Y \Phi^C ( Q \rho X Q) ) \\
&=&  \tr \Big(Y \big( \sum_{i,j} \tr( Q \rho X Q V_j^* V_i) \big) \kb{i}{j} \Big) \\
&=&  \sum_{i,j} \tr (\rho X Q V_j^* V_i Q)  \tr (Y\kb{i}{j}),
\end{eqnarray*}
from which we can conclude from Eqs.~(\ref{testablecond}) that these two quantities are equal. As $X, Y, \rho$ were arbitrary, it follows that $[\mathcal P_Q \circ (\Phi^C)^\dagger(Y),X]=0$ for all $X\in \mathcal A$ and hence $\mathcal A$ is private for $\Phi^C$ with respect to $Q$.

For the converse direction, suppose that $\mathcal A$ is private for $\Phi^C$ with respect to $Q$. Then $\mathcal P_Q \circ (\Phi^C)^\dagger (\mathcal L(\mathcal H_C)) \subseteq \mathcal A^\prime$, and so for all $X\in \mathcal A$, $\rho\in\mathcal L(\mathcal H)$, $Y\in \mathcal L(\mathcal H_C)$ we have
\[
\tr(Q (\Phi^C)^\dagger(Y) Q X \rho)  = \tr( X Q (\Phi^C)^\dagger(Y) Q \rho),
\]
and hence from the above calculations that
\[
\sum_{i,j} \tr (\rho Q V_j^* V_i Q X)  \bra{j} Y \ket{i} = \sum_{i,j} \tr (\rho X Q V_j^* V_i Q)  \bra{j} Y \ket{i}.
\]
To conclude the proof, we now fix a pair $i_0,j_0$ and apply this identity with $Y = \kb{j_0}{i_0}$ to obtain
\[
\tr (\rho Q V_{j_0}^* V_{i_0} Q X) = \tr (\rho X Q V_{j_0}^* V_{i_0} Q),
\]
which holds for all $\rho$ and $X$. Thus it follows that $[QV_{j_0}^* V_{i_0} Q, X]=0$ for all $X\in \mathcal A$, and hence by Lemma~\ref{testable} we have that $\mathcal A$ is correctable for $\Phi$ with respect to $Q$, and the result follows.
\end{proof}

\section{Complementary Operator Structures}

Operator structures have previously been identified that describe quantum error correction; for instance, multiplicative domains for channels and certain generalizations of them were shown to characterize operator and standard quantum error correction as part of the early expanded work on subsystem codes \cite{c-j-k,johnston}. Below we shall briefly review these structures and then extend the correspondence to OAQEC.

First though, we will identify operator structures that characterize private (subspaces, subsystems, and) algebras. We begin with a simple observation of
an elementary connection between the null space of a quantum channel and the sets of states that it privatizes. Let $\Phi$ be a channel from $M_n$ to $M_m$ and let $S$ be the null space of $\Phi$.  If $\rho_1$ and $\rho_2$ are $n\times n$ density matrices, then $\Phi(\rho_1)=\Phi(\rho_2)$ if and only if $\rho_1-\rho_2\in S$.
This observation suggests that nullspaces of channels can be used to describe privacy, and indeed this is the case.

\begin{lemma}\label{privlemma}
Let $\mathcal H$ be a Hilbert space and let $Q$ be a projection on $\mathcal H$. Given a channel $\Phi : \mathcal L(\mathcal H) \rightarrow \mathcal L(\mathcal H)$ with $\Phi(\rho) = \sum_i V_i \rho V_i^*$, the following commutants inside $\mathcal L(\mathcal H)$ coincide:
\begin{equation}\label{privlemmaeqn}
\{ Q V_j^* V_i Q \}^\prime_{i,j} =  \big( \big( \ker (\Phi^C \circ \mathcal P_Q) \big)^\perp  \big)^\prime ,
\end{equation}
where orthogonality is with respect to the trace inner product.
\end{lemma}

\begin{proof}
By direct calculation using the form of the complementary map, we have
\[
(\Phi^C \circ \mathcal P_Q)(X)  = \Phi^C(QXQ) = \sum_{i,j} \tr(X Q V_j^* V_i Q) \, \kb{i}{j}.
\]
Hence, $\ker (\Phi^C \circ \mathcal P_Q) = (\spn \{ Q V_j^* V_i Q \}_{i,j})^\perp$.
\end{proof}

\begin{theorem}\label{privkerthm}
Let $\mathcal A$ be a subalgebra of $\mathcal L(Q \mathcal H)$, for some Hilbert space $\mathcal H$ and projection $Q\in \mathcal L(\mathcal H)$. Let $\Phi$ be a channel on $\mathcal H$ with complementary channel $\Phi^C$. Then $\mathcal A$ is private for $\Phi$ with respect to $Q$ if and only if $\mathcal A$ is contained inside $( ( \ker (\Phi \circ \mathcal P_Q) )^\perp  )^\prime$.
\end{theorem}

\begin{proof}
This can be proved by combining Theorem~\ref{testable} and Proposition~\ref{perfectcase} with Lemma~\ref{privlemma}, and use the fact (see chapter~6 of \cite{holevo}) that $(\Phi^C)^C$ is isometrically equivalent to $\Phi$.
\end{proof}

We can explicitly connect these private structures with the corresponding structures from error correction, the subject of which we now turn. For brevity we shall consider the correctable/private ($Q=I$) case. 

\begin{definition}
The {\bf multiplicative domain}, $\mathcal{M}(\Phi)$, of a channel $\Phi : \mathcal L(\mathcal H) \rightarrow \mathcal L(\mathcal H)$ is the set (in fact an algebra) given by:
$$\mathcal{M}(\Phi)=\{A \in \mathcal L(\mathcal H): \Phi(AX) = \Phi(A)\Phi(X);  \Phi(XA) = \Phi(X)\Phi(A) \  \forall X\},$$
where $X$ is taken from $\mathcal L(\mathcal H)$. The multiplicative domain is the largest set on which the restriction of $\Phi$ is a $\ast$-homomorphism (i.e., a representation).

Given a subalgebra $\mathcal A \subseteq \mathcal L(\mathcal H)$ and a representation $\pi : \mathcal A \rightarrow \mathcal L(\mathcal H)$ (that is, $\pi$ is linear, 
$\pi(AB) = \pi(A)\pi(B)$ and $\pi(A)^* = \pi(A^*)$ for all $A,B\in\mathcal A$), we may also define {\bf generalized multiplicative domains} as follows:
$$\mathcal{M}_{\pi}(\Phi) = \{A \in \mathcal A: \Phi(AX) = \pi(A)\Phi(X);  \Phi(XA) = \Phi(X)\pi(A) \  \forall X\in \mathcal L(\mathcal H)\}.$$
\end{definition}

The following quantum error correction result was established for subsystem codes in \cite{johnston}, and here we show that it extends to OAQEC. Our proof is built on techniques from \cite{johnston} and error correction constructions from \cite{bkk1,bkk2}.

\begin{theorem}\label{mdpicorrectable}
Let $\mathcal A$ be a subalgebra of $\mathcal L(\mathcal H)$ and let $\Phi$ be a channel on $\mathcal L(\mathcal H)$. Then $\mathcal A$ is correctable for $\Phi$ if and only if  $\mathcal A = \mathcal M_\pi (\Phi)$ for some representation $\pi : \mathcal A \rightarrow \mathcal L(\mathcal H)$.
\end{theorem}

\begin{proof}
First suppose $\mathcal A = \mathcal M_\pi (\Phi)$. So $\Phi(AX) = \pi(A)\Phi(X)$ for all $A\in \mathcal{A}$ and $X \in \mathcal L(\mathcal H)$,  $\pi(AB) = \pi(A)\pi(B)$ and $\pi(A)^* = \pi(A^*)$ for all $A,B\in \mathcal{A}$.
Then, since $\Phi$ is trace-preserving, we have
$$\mathrm{Tr}(AX) = \mathrm{Tr}(\Phi(AX)) = \mathrm{Tr}(\pi(A)\Phi(X)) = \mathrm{Tr}(\Phi^{\dagger}(\pi(A))X)$$ for all $X\in \mathcal L(\mathcal H)$ and hence $\Phi^{\dagger}(\pi(A)) = A$ for all $A\in \mathcal{A}$.

Now, let $A\in \mathcal A$ and observe since $\pi$ is a homomorphism and $\Phi^{\dagger}(\pi(A)) = A$, we have 
\begin{eqnarray*}
\Phi^{\dagger}(\pi(A)\pi(A^*)) - A\Phi^{\dagger}(\pi(A^*)) &-& \Phi^{\dagger}(\pi(A))A^* + AA^* \\ &=& AA^* - AA^* - AA^* + AA^* \\ &=& 0.
\end{eqnarray*}
However, observe this quantity is also equal to (recalling $\Phi^\dagger(Y) = \sum_i V_i^* Y V_i$ and $\Phi^{\dagger}$ is unital since $\Phi$ is trace-preserving) the following sum when fully expanded:  
$$\sum_i (V_i^* \pi(A) - A V_i^*)(V_i^*\pi(A) - AV_i^*)^* = \sum_i (V_i^*\pi(A) - A V_i^*)(\pi(A^*) V_i - V_i A^*).$$ 
Hence, it follows that each term in this sum is $0$, and so we must have (also using the fact that $\mathcal A$ is a self-adjoint set)
\[
V_i^*\pi(A) = A V_i^*  \quad \mathrm{and}  \quad \pi(A) V_j  = V_j A.
\]
Multiply the first equation on the right by $V_j$ and the second equation on the left by $V_i^*$ to obtain
\[
A V_i^*V_j = V_i^*\pi(A)V_j = V_i^*V_j A,
\]
and thus we have shown that  $\mathcal{A} \subseteq \{V_i^*V_j\}'$ and $\mathcal A$ is correctable for $\Phi$.

For the converse implication, assume that $\mathcal{A}\subseteq \{V_i^*V_j\}'$.  Let $R = \Phi(I) = \sum_i V_i V_i^*$; notice that
\begin{equation}
\Phi(A)R = \sum_{i,j} V_i A V_i^*V_j V_j^* = \sum_{i,j} V_i V_i^*V_j A V_j^* = R \Phi(A),
\end{equation}
and so $\Phi(A)$ commutes with any power of $R$ for all $A\in \mathcal A$.
Next, observe that \begin{equation}\label{almostmultiplicative1a}\Phi(A)\Phi(X)= \sum_{i,j} V_i A V_i^*V_j X V_j^* = \sum_i V_iV_i^*\sum_j V_j AX V_j^* = R \Phi(AX)\end{equation} and similarly,
\begin{equation}\label{almostmult2a}\Phi(X)\Phi(A) = \Phi(XA)R.\end{equation}
If $R$ is invertible, we then obtain
$$ \Phi(AX) = R^{-1}\Phi(A)\Phi(X) = R^{-1/2}\Phi(A)R^{-1/2}\Phi(X);$$
and,
$$\Phi(XA) = \Phi(X)\Phi(A)R^{-1} = \Phi(X)R^{-1/2}\Phi(A)R^{-1/2}.$$
Defining $\pi(A) = R^{-1/2}\Phi(A)R^{-1/2}$ we see that the above can be written as
\[
\Phi(AX) = \pi(A)\Phi(X) \quad \mathrm{and} \quad \Phi(XA) = \Phi(X)\pi(A),
\]
and we note that for any $A,B\in \mathcal{A}$,
\begin{align*}\pi(A)\pi(B) &= R^{-1/2}\Phi(A)R^{-1}\Phi(B)R^{-1/2} \\
& = R^{-1/2}\Phi(A)\Phi(B)R^{-3/2} \\
& = R^{-1/2}\Phi(AB)R^{-1/2}\\
& = \pi(AB)\end{align*}
where we have used the fact that $\Phi(B)$ commutes with all powers of $R$ and that $\Phi(A)\Phi(B) = \Phi(AB)R$. It follows that $\mathcal A = \mathcal M_\pi (\Phi)$ as required.

If $R$ is not invertible, we note that $\ker(R) = \cap_i \ker(V_i^*)$ and so if $R \ket{\psi} = 0$, then $\Phi(X)\ket{\psi} = \sum_i V_i X V_i^*\ket{\psi} = 0$ and $\bra{\psi} \Phi(X) = \sum_i \bra{\psi} V_i X V_i^* = 0$ as well. Hence, if we write $\mathcal{V} = \ker (R)$ and split $\mathcal H = \mathcal{V}^{\perp} \oplus \mathcal{V}$, according to this decomposition $R = Q \oplus 0$ with $Q$ invertible, and $\Phi(X) = \Psi(X) \oplus 0$. Hence, returning to Eq.~(\ref{almostmultiplicative1a}) we see that it can be written as
$$(\Psi(A)\Psi(X))\oplus 0 = (Q\Psi(AX))\oplus 0,$$ and if we multiply by $Q^{-1}\oplus 0$ we get
$$(Q^{-1}\Psi(A)\Psi(X))\oplus 0 = \Psi(AX) \oplus 0$$ and hence
\[
R^+\Phi(A)\Phi(X) = \Phi(AX)
\]
where $R^+$ is the pseudo-inverse of $R$. Similarly, we can do the same for Eq.~(\ref{almostmult2a}) and let $\pi(A) = (R^+)^{-1/2}\Phi(A)(R^+)^{-1/2}$ to get the desired result that $\mathcal A = \mathcal M_\pi (\Phi)$, and this completes the proof.
\end{proof}

\begin{remark}
Observe from the start of the above proof that any correctable algebra for $\Phi$ is contained in the range of $\Phi^\dagger$. (This was also observed from a different perspective in \cite{bkk2}.) We will use this fact in the next section.
\end{remark}

\begin{example}
As an illustration of this correspondence, consider a 4-qubit channel $\Phi$ that models noise given by the possibility of independent bit flips on the first three qubits, and so $\Phi$ has four Kraus operators (normalized with probabilities) $I$, $X_1 = X\otimes I \otimes I \otimes I$, and $X_2$, $X_3$ similarly defined with $X$ the Pauli bit flip operator ($X\ket{0}=\ket{1}$, $X\ket{1}=\ket{0}$). Consider the orthogonal single-qubit subspaces $\mathcal C_0 = \spn\{ \ket{0000},\ket{1111}\}$, $\mathcal C_1 = \spn \{ \ket{0001}, \ket{1110}\}$. Each of these subspaces is easily seen to be individually correctable for $\Phi$, but more than this, one can check that the hybrid algebra code defined by the subspaces, namely $\mathcal A = \mathcal L(\mathcal C_0) \oplus \mathcal L(\mathcal C_1)$, is correctable for $\Phi$. The theorem tells us therefore that the code algebra coincides with a generalized multiplicative domain for $\Phi$, $\mathcal A = \mathcal M_\pi (\Phi)$, and indeed, the proof also gives a recipe for constructing the representation: in this case, the representation $\pi : \mathcal A \rightarrow \mathcal L(\mathcal H)$ is implemented by the four Kraus operators $\{ P_{\mathcal C}, P_{\mathcal C_i}X_i , i=1,2,3\}$, where $\mathcal C = \mathcal C_1 \oplus \mathcal C_2$ and $P_{\mathcal C_i} =  X_i P_{\mathcal C} X_i^*$.
\end{example}

Combining the previous result with Theorem~\ref{privkerthm}, the complementary operator structure relationship is revealed as follows.

\begin{corollary}\label{corrpriv}
Let $\Phi$ be a channel on $\mathcal L(\mathcal H)$, and let $\pi$ be a representation associated with a correctable algebra for $\Phi$ (or equivalently a private algebra for $\Phi^C$). Then we have
\[
 \mathcal M_\pi (\Phi) = \big( \big( \ker (\Phi^C ) \big)^\perp  \big)^\prime .
\]
\end{corollary}

\begin{proof}
The forward inclusion follows from Theorem~\ref{privkerthm}, and the opposite inclusion follows from the second half of the proof of Theorem~\ref{mdpicorrectable}. 
\end{proof}

\begin{remark}
The simplest illustration of this relationship comes from the extreme case of a correctable/private pair, the case with $\Phi = \mbox{id}$ the identity channel on $\mathcal L(\mathcal H)$. Here $\Phi^C(\rho) = \tr(\rho)$ is the completely depolarizing channel, $\pi = \mbox{id}$, $Q=I$, and $\mathcal M_\pi (\Phi) = \mathcal L(\mathcal H)$. Moreover, $\ker \Phi^C$ is the operator subspace of trace-zero matrices, which is the trace-orthogonal complement of the identity operator $I$, and hence $(\ker\Phi^C)^\perp$ is the set of scalar multiples of the identity, with commutant equal to $\mathcal L(\mathcal H)$ as given by the result. For the example above, the specific form of the complement is not as straightforward, nevertheless the result yields information on it; namely, in that case $\ker \Phi^C$ can be explicitly computed via the relation $\mathcal L(\mathcal C_0) \oplus \mathcal L(\mathcal C_1)= ( ( \ker \Phi^C )^\perp )^\prime$.

We further note it would be interesting to extend this result to the general projection $Q$ case. This should be possible but there are some technical issues to overcome on how to define the multiplicative domains in that case. 
\end{remark}

\subsection{The Special Case of Unital Channels} We finish this section by continuing the analysis in the distinguished special case of unital channels ($\Phi(I) = I$). Many physically relevant channels satisfy this extra condition, such as the previous example. The relevant structures, in particular the multiplicative domains, have an especially nice characterization.

If $\Phi : \mathcal A \rightarrow \mathcal B$ is a completely positive and unital map between two algebras, then Choi \cite{choi1} proved that $\mathcal M(\Phi)$ has the following internal description:
\[
\mathcal M(\Phi) = \{ A\in \mathcal A \, : \, \Phi(A)^*\Phi(A) = \Phi(A^* A), \,\, \Phi(A) \Phi(A)^* = \Phi(AA^*)   \}.
\]
When trace preservation is added, so $\Phi(\rho) = \sum_i V_i \rho V_i^*$ is a unital channel, the fixed point theory for such maps \cite{kribsfixed} can be built upon to prove \cite{kribs-spekkens,c-j-k,miza} that $\mathcal M(\Phi)$ is equal to the commutant of the operators $V_i^*V_j$, it encodes all unitarily correctable algebras for $\Phi$, and the unital channel $\Phi^\dagger$ acts as a recovery operation; in terms of operator structures this is stated as:
\[
\mathcal M(\Phi) = \mathrm{UCC}\,(\Phi) = \{ V_i^* V_j \}^\prime = \mathrm{Fix}\, (\Phi^\dagger \circ \Phi).
\]

We can thus state the following result based on the above.

\begin{corollary}
If $\Phi(\rho) = \sum_i V_i \rho V_i^*$ is a unital channel, then
\[
\mathcal M(\Phi)=\{ V_i^* V_j \}^\prime = (\ker(\Phi^C)^{\perp})' .
\]
\end{corollary}

It is clear that the null space of a channel and its multiplicative domain cannot both be large.  This relationship can thus be quantified by the following result in the unital case.

\begin{corollary} Let $\Phi$ be a unital quantum channel on $\mathcal L(\mathcal H)$.  Then $$\dim(\mathcal M(\Phi))+\dim(\ker(\Phi))\le (\dim \mathcal H)^2$$ with equality if and only if $\Phi^\dagger \circ \Phi $ is a projection.
\end{corollary}

\begin{proof} From the discussion above we know $X\in \mathcal M(\Phi)$ if and only if $X$ is an eigenvector of $\Phi^\dagger\circ \Phi $ corresponding to the eigenvalue one.  Similarly $X\in \ker (\Phi)$ if and only if $X$ is an eigenvector of $\Phi^\dagger\circ \Phi $ corresponding to the eigenvalue zero.  The inequality follows from this.  The inequality becomes an equality if and only if all of the eigenvalues of the positive semidefinite operator $\Phi^\dagger\circ \Phi $ are zero and one, which occurs if and only if $\Phi$ is a projection. \end{proof}

We note that since $\mathcal M(\Phi)$ is a unital von Neumann subalgebra of some $M_n$, the projection onto $\mathcal M(\Phi)$ would be the trace-preserving conditional expectation onto $\mathcal M(\Phi)$, which is the unique channel $\Phi_{\mathcal A}$ satisfying:
\begin{enumerate}
\item $\Phi_{\mathcal{A}}(A)=A \quad \forall A\in\mathcal{A}$
\item $\Phi_{\mathcal{A}}(A_1XA_2)=A_1\Phi_{\mathcal{A}}(X)A_2 \quad \forall A_1,A_2\in\mathcal{A}, \quad \forall X\in M_n$
\end{enumerate}
Among all unital quantum channels with a given multiplicative domain $\mathcal{A}$, the trace-preserving conditional expectation onto $\mathcal{A}$ has the largest possible nullspace.

We conclude this section by deriving some relations on the behaviour of the complementary channel and a channel's multiplicative domain in the unital case.

 \begin{proposition}
 Let $\Phi$ be a unital channel on $\mathcal L(\mathcal H)$. Then  for all $X\in \mathcal L(\mathcal H)$ and $A\in \mathcal M(\Phi)$, we have
 \[\Phi^C(AX)=\Phi^C(XA). \]
 \end{proposition}

\begin{proof}
We have $A\in \mathcal M(\Phi)$ if and only if $AV_i^*V_j=V_i^*V_jA$ for all $i,j$. So $\tr(AX V_i^*V_j) = \tr(X V_i^*V_j A) =  \tr(XA V_i^*V_j)$, and hence
\[
\Phi^C(AX) = \sum_{ij}\tr(AX V_i^*V_j) \kb{i}{j} =  \sum_{ij}\tr(XA V_i^*V_j) \kb{i}{j} =\Phi^C(XA).
\]
\end{proof}

This result has some interesting consequences.

\begin{corollary}
Let $\Phi$ be a unital channel on $\mathcal L(\mathcal H)$. If $\mathcal M(\Phi)$ is a von Neumann algebra factor then $\Phi^C (\mathcal M(\Phi))=\mathbb{C}I$.
\end{corollary}

\begin{proof}  Suppose $A\in \mathcal M(\Phi)$ with $tr(A)=0$.  Since $\mathcal M(\Phi)$ is a von Neumann algebra factor and hence isomorphic to a matrix algebra, there exists $X,Y\in \mathcal M(\Phi)$ such that $A=XY-YX$ and thus $\Phi^C(A)=\Phi^C(XY)-\Phi^C(YX)=0$.  Since every element of $\mathcal M(\Phi)$ is the sum of a trace zero element of $\mathcal M(\Phi)$ and a multiple of the identity, the result follows.
\end{proof}

\begin{corollary}\label{commutcor}
Let $\Phi$ be a unital channel. Then the set $\Phi^C(\mathcal M(\Phi))$ commutes with the set $\Phi^C(\mathcal M(\Phi^C))$.
\end{corollary}

\begin{proof}
Let $A\in \mathcal M(\Phi )$ and $X\in \mathcal M(\Phi^C )$.  Then from the previous result and the multiplicative domain definition, we have $\Phi^C(A)\Phi^C(X)=\Phi^C(AX)=\Phi^C(XA)=\Phi^C(X)\Phi^C(A)$.
\end{proof}

\begin{remark}
Regarding the generalized multiplicative domains, in the case that $\Phi$ is a unital channel, we have $\mathcal{M}(\Phi) = \{V_i^*V_j\}'$. Hence in this case, all generalized multiplicative domains associated with unital subalgebras lie inside the actual multiplicative domain. Also note that if $\Phi(\rho) = \sum_i V_i \rho V_i^*$ is a channel such that the $V_i = (\dim\mathcal H)^{-1/2} U_i$, with $\{ U_i \}$ a set of unitaries that are mutually orthogonal in the trace inner product, then of course $\Phi$ is unital. But also observe that $\Phi^C$ is unital as well:
\[
\Phi^C(I) = \sum_{ij}\tr(V_i^*V_j) \kb{i}{j} =   \sum_{i}\tr(V_i^*V_i) \kb{i}{i} = \sum_i \kb{i}{i} = I .
\]
In particular, in the above results the roles of $\Phi$ and $\Phi^C$ can be interchanged.

An interesting example of this arises when $\Phi$ is the conditional expectation onto the diagonal matrices.  In this case, $\Phi^C=\Phi$ which means that $\Phi(\mathcal M(\Phi))=\Phi^C(\mathcal M(\Phi))=\Phi^C(\mathcal M(\Phi^C))$ must be contained in an abelian subalgebra by Corollary \ref{commutcor}.  An easy calculation show that this is indeed the case with $\Phi(\mathcal M(\Phi))$ being the algebra of diagonal matrices.
\end{remark}

\section{Operator Algebra Inequalities and the Correction vs Privacy Trade-Off}

In this section we build on the analysis above to further quantify complementarity, through inequalities determined by the sizes of the relevant operator algebras corrected or privatized by channels. To simplify the presentation we shall use matrix notation for the algebras and we will focus on the original basic notion of privacy, where an algebra is mapped to a single state: Given a channel $\Phi:M_n\rightarrow M_m$ and subalgebra $\mathcal A$, we suppose there is a density operator $\rho$ such that
$$\Phi(A) = \mathrm{Tr}(A) \rho \quad\quad \forall A\in \mathcal A .$$
In such a situation, we shall say $\mathcal A$ is {\it privatized to a state} by $\Phi$.

Up to unitary equivalence our algebras, say contained in $M_n$, have the form
$$\mathcal{A} = \bigl(\oplus_{k=1}^N I_{m_k}\otimes M_{n_k}\bigr) \oplus 0_K, $$ with $\sum_k m_kn_k +K = n$. The commutant of $\mathcal{A}$ is, up to the same unitary similarity,
$$\mathcal{A}' = \bigl( \oplus_{k=1}^N M_{m_k} \otimes I_{n_k}\bigr) \oplus M_K.$$
Note that $$\mathcal{A} \cap \mathcal{A}' = \bigl(\oplus_{k=1}^N \C I_{m_k}\otimes I_{n_k}\bigr) \oplus 0_K$$ and that the dimension of this algebra (as an operator space) is
$$\mathrm{dim}(\mathcal{A}\cap \mathcal{A}') = N.$$
Finally, we note that $\mathcal{A}$ has a largest central projection $P_{\mathcal A}$; up to the same unitary similarity as before,
$$P_{\mathcal A} = \bigl( \oplus_{k=1}^N I_{m_k}\otimes I_{n_k}\bigr) \oplus 0_K.$$ This projection satisfies:
$P_{\mathcal A} A = A P_{\mathcal A} = A$ for all $A \in \mathcal{A}$. Note that $\mathcal{A}$ is a unital algebra if and only if $P_{\mathcal A}=I_n$ if and only if $K=0$. We shall also focus on the unital algebra case in this section.

The next two results refer to the notion of quasiorthogonal algebras. We point the reader to \cite{ljklp} for more on the notion and its connections with privacy.

\begin{definition} Two unital subalgebras $\mathcal{A}, \mathcal{B}\subseteq M_n$ are said to be {\bf quasiorthogonal} if $\tr(AB)=n^{-1}\tr(A)\tr(B)$ for all $A\in\mathcal{A} $ and all $B\in\mathcal{B }$.
\end{definition}

\begin{lemma}\label{quasilem} An algebra $\mathcal{B}$ is privatized to a state by $\Phi$ if and only if $$\tr(B\Phi^{\dagger}(X)) = \tr(B)\tr(\rho X)$$ for all $X\in M_n$. If $\mathcal{B}$ is unital, this is equivalent to the quasiorthogonality of $\mathcal{B}$ and $\mathrm{range}(\Phi^{\dagger})$.
\end{lemma}

\begin{proof} The first statement is trivial. For the second, if $I\in \mathcal{B}$, then $\Phi(I) = n\rho$ and so $\rho = n^{-1}\Phi(I)$. Hence we have
$\tr(B\Phi^{\dagger}(X)) = n^{-1}\tr(B)\tr(\Phi(I)X) = n^{-1}\tr(B)\tr(\Phi^{\dagger}(X)).$
\end{proof}

Recall in the previous section we saw that a correctable algebra for $\Phi$ must lie in the range of $\Phi^{\dagger}$.
 The example of $\Phi_\mathcal{A}$, the trace-preserving conditional expectation onto a unital algebra $\mathcal{A}$ is instructive. Since $\Phi_\mathcal{A}= \Phi_\mathcal{A}^{\dagger}$, the range of $\Phi_\mathcal{A}^{\dagger}$ is $\mathcal{A}$ which clearly is a correctable algebra for $\Phi_\mathcal{A}$ as a subalgebra of its fixed point set. It was noted by Petz that two subalgebras $\mathcal{A}$ and $\mathcal{B}$ are quasiorthogonal if and only if $\Phi_\mathcal{A}(B)$ is a multiple of the identity for all $B\in \mathcal{B}$ \cite[Theorem 3]{petz2007complementarity}. Hence a unital algebra $\mathcal{B}$ will be privatized by $\Phi_\mathcal{A}$ if and only if it is quasiorthogonal to $\mathcal{A}$. This observation coupled with Lemma~\ref{quasilem} gives us the following characterization of correctable/privatized algebra pairs.

 \begin{theorem} \label{quasithm} Let $\mathcal{A}$ and $\mathcal{B}$ be unital subalgebras of $M_n$.  Then there exists a quantum channel that corrects $\mathcal{A}$ and privatizes $\mathcal{B}$ to a state if and only if $\mathcal{A}$ and $\mathcal{B}$ are quasiorthogonal. \end{theorem}

 We can thus prove the following.

\begin{corollary}\label{submultiplicativedims} Suppose $\mathcal{A}$ is a correctable algebra for a channel $\Phi:M_n\rightarrow M_m$, and $\mathcal{B}$ is a unital algebra privatized to a state by $\Phi$. Then $$\mathrm{dim}(\mathcal{A})\mathrm{dim}(\mathcal{B})\leq n^2.$$
\end{corollary}

\begin{proof} 

Since $\mathcal{A}$ and $\mathcal{B}$ are quasiorthogonal by Theorem \ref{quasithm}, for all $A\in \mathcal{A}$ and $B\in \mathcal{B}$, we have
$\tr(AB) = n^{-1}\tr(A)\tr(B)$.  Let $\{A_i\}_{i=1}^{d_1}$ and $\{B_i\}_{i=1}^{d_2}$ be orthonormal bases (in the trace inner product) for $\mathcal{A}, \mathcal{B}$ respectively. Next form the set $\{A_iB_j\}_{i,j=1}^{d_1,d_2}$ which has $d_1d_2 = \dim (\mathcal{A})\dim(\mathcal{B})$ elements and observe that
\begin{align*}\tr(A_iB_j(A_kB_l)^*) & = \tr(A_k^*A_iB_j B_l^*) \\
& = n^{-1}\tr(A_k^*A_i)\tr(B_jB_l^*) \\
& = n^{-1}\delta_{ik}\delta_{jl}\end{align*} and so $\{A_iB_j\}_{i,j=1}^{d_1,d_2}$ is a set of mutually orthogonal matrices in $M_n$, and so must have dimension at most $n^2$.
\end{proof}

\begin{example}
As a simple example of a channel and algebras that saturate this inequality, consider an $N$-qubit system (so $n = 2^N$) and noise given by a channel $\Phi$ that completely depolarizes the first $k$ qubits and leaves the final $N-k$ qubits untouched. In this case, we have a correctable (noiseless in fact) algebra $\mathcal A$ that is unitarily equivalent to $M_{2^{N-k}}$ and a private algebra $\mathcal B$ that is privatized to the maximally mixed state of the first $k$ qubits and is unitarily equivalent to $M_{2^k}$. Here we thus have: $\dim (\mathcal A) \dim (\mathcal B) = 2^{2(N-k)} 2^{2k} = 2^{2N} = n^2$.
\end{example}

\begin{remark}
Theorem \ref{quasithm} suggests a way to quantify the complimentarity relations of two subalgebras. Indeed, in \cite{Wein}, a quantity $(c(\mathcal{A},\mathcal{B}))$ was defined for two unital subalgebras of a matrix algebra in terms of the trace of the composition of conditional expectations onto each of the subalgebras. It follows that $1\leq c(\mathcal{A},\mathcal{B})\leq \min\{\text{dim}(\mathcal{A}),\text{dim}(\mathcal{B})\}.$
Moreover, $\mathcal{A}$ and $\mathcal{B}$ are quasiorthogonal if and only if $c(\mathcal{A},\mathcal{B})=1$.
Now using this quantity one can measure how far away two subalgebras are from being quasiorthogonal. Via Theorem \ref{quasithm}, this allows us to see quantitatively how far away a pair of subalgebras are from being complimentary to each other, and we suggest this warrants further investigation.
\end{remark}

The following example illustrates the necessity of the unitality condition in Corollary~\ref{submultiplicativedims}.

\begin{example}\label{counterexample}  Let $\Phi:M_3\rightarrow M_3$ be the channel whose Kraus operators are $\begin{pmatrix} 1 & 0 & 0 \\ 0 & 1 & 0\\ 0 & 0 & 1 \end{pmatrix}, \begin{pmatrix} 0 & 0 & 0 \\ 0 & 0 & 0 \\ 0 & 1 & 0 \end{pmatrix}$, which acts by
$$\begin{pmatrix} a & b & c \\ d & e & f \\ g & h &i \end{pmatrix} \mapsto \begin{pmatrix} a & b & 0 \\ d & e + i & 0 \\ 0 & 0 & 0 \end{pmatrix}.$$
Observe that the algebra $\mathcal{A} = M_2\oplus 0$ is correctable for $\Phi$, since on this algebra, $\Phi$ acts as the identity. Moreover, $\mathcal{B} = 0\oplus M_2$ is private for this algebra, since
$$\Phi: \begin{pmatrix} 0 & 0 & 0 \\ 0 & a & b \\ 0 & c & d \end{pmatrix} \mapsto (a+d) \begin{pmatrix} 0 & 0 & 0 \\ 0 & 1 & 0 \\ 0 & 0 & 0 \end{pmatrix}.$$
Both $\mathcal{A}$ and $\mathcal{B}$ are of dimension $4$, and so we have $\mathrm{dim}(\mathcal{A})\mathrm{dim}(\mathcal{B}) = 16 \nleq 9$.
\end{example}

We next relate the commutants of correctable/private algebra pairs.

\begin{corollary}\label{dimineq} If $\mathcal{A}$ and $\mathcal{B}$ are unital algebras correctable and privatized to a state by $\Phi : M_n \rightarrow M_m$ respectively, then $$\dim(\mathcal{B}) \leq \dim(\mathcal{A}')\quad \mbox{and} \quad \dim(\mathcal{A})\leq \dim(\mathcal{B}').$$
\end{corollary}

\begin{proof} The result follows from the theorem above and the fact that unital algebras satisfy
$n^2 \leq \dim(\mathcal{A})\dim(\mathcal{A}').$
\end{proof}

Note that Example \ref{counterexample} again serves as a reminder that unitality is necessary: in that example, $\mathrm{dim}(\mathcal{A}') = \mathrm{dim}(\mathcal{B}') = 3$ while $\mathrm{dim}(\mathcal{A}) = \mathrm{dim}(\mathcal{B}) = 4$.

We can also make a statement on the internal structures of correctable and private algebra pairs.

\begin{corollary} Suppose $\mathcal{A}$ and $\mathcal{B}$ are unital algebras correctable and privatized to a state by $\Phi$ respectively. The following are true:
\begin{enumerate}
\item If $\mathcal{A}$ contains a maximal abelian subalgebra, then $\mathcal{B}$ cannot, unless $\mathcal{A}$ is itself a maximal abelian subalgebra, in which case so is $\mathcal{B}$.
\item If $\mathcal{B}$ contains a maximal abelian subalgebra, then $\mathcal{A}$ cannot, unless $\mathcal{B}$ is itself a maximal abelian subalgebra, in which case so is $\mathcal{A}$.
\end{enumerate}
\end{corollary}

\begin{proof}
This follows from the observation that the commutant of an algebra containing a maximal abelian subalgebra is abelian, and has dimension less than $n$, while the dimension of $\mathcal{A}$ itself is greater than $n$. The inequality in Corollary~\ref{dimineq}, and a consideration of the equality condition, give the result.
\end{proof}

Further recall from above that all unital correctable algebras $\mathcal{A}$ satisfy
$\mathcal{A} \subseteq \{V_i^*V_j\}'$, and hence
$\{V_i^*V_j\}'' \subseteq \mathcal{A}'.$ Hence, the smallest possible commutant we can put on the right side of the inequality from Corollary~\ref{dimineq} is $\dim (\{V_i^*V_j\}'')$, giving us the following result. (And recall by the von Neumann double commutant theorem, $\{ V_i^* V_j \}^{''}$ is equal to the algebra generated by the operators $V_i^* V_j$.)

\begin{corollary} If $\mathcal{B}$ is a unital algebra privatized to a state by $\Phi$, then $$\dim(\mathcal{B}) \leq \dim(\{V_i^*V_j\}'').$$
\end{corollary}

\begin{remark}
We see then that, at least for unital algebras and privatizing to states, these inequalities exhibit an explicit and concrete trade-off between privacy and correction: if a large algebra is correctable for $\Phi$, the size of the largest privatized algebra is constrained to be small, and vice-versa.
\end{remark}

Recalling the Kraus operator description of $\Phi^C$, we finish by analyzing what happens for the complement when $\mathcal{A}$ is correctable for $\Phi$.

\begin{proposition}\label{comppriv} If $\mathcal{A}$ is correctable for $\Phi$, then 
\[
\mathrm{rank} (\Phi^C\bigl|_{\mathcal{A}}) \leq \dim ( \mathcal A \cap \mathcal A^\prime).
\] 
\end{proposition}

\begin{proof} We have that $\Phi^C(A)_{ij} = \tr (V_j^*V_i A)$; using the fact that $A = QAQ$ where $Q$ is the largest central projection in $\mathcal{A}$, the above becomes
$$\tr(V_j^*V_i QAQ) = \tr (QV_j^*V_i QA);$$ recalling the previous section, $QV_j^*V_i Q \in \mathcal{A}'$. If $\mathcal{A}$ is unitarily equivalent to $\bigl(\oplus_{k=1}^N I_{m_k}\otimes M_{n_k}\bigr) \oplus 0_K$ then $\mathcal{A}'$ is unitarily equivalent to $\bigl(\oplus_{k=1}^N M_{m_k}\otimes I_{n_k}\bigr)\oplus M_K$, and so for any $A$ we have that $A$ is unitarily equivalent to $\bigl(\oplus_{k=1}^N I_{m_k}\otimes A_k\bigr)\oplus 0_K$ and $QV_j^*V_iQ$ is unitarily equivalent to $\bigl(\oplus_{k=1}^N V^{(k)}_{ji}\otimes I_{n_k}\bigr)\oplus 0_K$; and so
$$\tr (QV_j^*V_i QA) = \tr \biggl(\bigl(\oplus_{k=1}^N V^{(k)}_{ji}\otimes A_k\bigr)\oplus 0_K\biggr) = \sum_{k=1}^N \tr(A_k)\tr(V_{ij}^{(k)}).$$
Hence $$\Phi^C(A) = \sum_{i,j} E_{ji}\biggl(\sum_{k=1}^N \tr (A_k)\tr (V_{ij}^{(k)})\biggr) = \sum_{k=1}^N \tr (A_k)\Lambda_k,$$ where $\Lambda_k = \sum_{i,j} \tr (V_{ij}^{(k)})E_{ji}$.
Clearly, $N=\dim(\mathcal{A}\cap \mathcal{A}')$ and is obviously an upper bound on the dimension of the range of $\Phi^C(\mathcal{A})$.
\end{proof}

\begin{remark}
Considering the general notion of privacy from \cite{cklt}, note that the only information preserved from a private algebra is the information stored in $\mathcal{A}\cap \mathcal{A}'$; as this is an abelian and hence unitarily diagonalizable algebra, all the information can be considered as simply classical probabilities by reading off the diagonal, and so no genuine quantum information survives. By the analysis in Proposition \ref{comppriv}, we see that if $\mathcal{A}$ is correctable for $\Phi$, all of $\mathcal{A}\setminus (\mathcal{A}\cap \mathcal{A}')$ is sent to $0$, and the image of $\mathcal{A}$ under $\Phi^C$ depends only on a compression of $\mathcal{A}$ to $\mathcal{A}\cap \mathcal{A}'$. Hence, by the reasoning in the paper \cite{cklt}, this counts as privatization: only classical information from the diagonal can survive.
\end{remark}

\section{Outlook}

The results of Section~3 give an expanded and deeper understanding of a key piece of quantum complementarity, namely the relationship between correctable and private algebras and codes for channels in quantum information, explicitly in terms of relevant operator structures. The previous section considered these structures and their joint sizes, and established a number of dimension inequalities that further quantify the correction and privacy trade-off for the special case of algebras privatized to quantum states. Further investigation of these relationships and inequalities is warranted. In particular, an expansion of the inequality analysis to general private algebras would be interesting, as would a deeper investigation on the quantification of complementarity via the operator structure identities derived here combined with the theory of quasiorthogonal algebras. Finally, due to practical considerations we have focussed here on the finite-dimensional case, but the notions of correctable and private algebras have been identified for general infinite-dimensional von Neumann algebras, and an exploration of extending the results presented here to that setting could be interesting. We plan to continue these investigations elsewhere.

\vspace{0.1in}

{\noindent}{\it Acknowledgements.} We are grateful to the referee for helpful comments. D.W.K. was partly supported by NSERC and a University Research Chair at Guelph. M.N. was partly supported by a University of Guelph International Student Scholarship, AIMS, and Mitacs. R.P. was partly supported by NSERC. M.R is supported by a postdoctoral fellowship in the Department of Pure Mathematics at the University of Waterloo.

\bibliography{privacy-mult}

\providecommand{\bysame}{\leavevmode\hbox to3em{\hrulefill}\thinspace}
\providecommand{\MR}{\relax\ifhmode\unskip\space\fi MR }
\providecommand{\MRhref}[2]{%
  \href{http://www.ams.org/mathscinet-getitem?mr=#1}{#2}
}
\providecommand{\href}[2]{#2}
\begin{thebibliography}{10}

\bibitem{ambainis}
Andris Ambainis, Michele Mosca, Alain Tapp, and Ronald de~Wolf, \emph{Private
  quantum channels}, IEEE Symposium on Foundations of Computer Science~(FOCS)
  (2000), 547--553.

\bibitem{bartlett2}
Stephen~D Bartlett, Patrick Hayden, and Robert~W Spekkens, \emph{Random
  subspaces for encryption based on a private shared {C}artesian frame},
  Physical Review A \textbf{72} (2005), 052329.

\bibitem{bartlett1}
Stephen~D Bartlett, Terry Rudolph, and Robert~W Spekkens,
  \emph{Decoherence-full subsystems and the cryptographic power of a private
  shared reference frame}, Physical Review A \textbf{70} (2004), 032307.

\bibitem{brs}
\bysame, \emph{Decoherence-full subsystems and the cryptographic power of a
  private shared reference frame}, Physical Review A \textbf{70} (2004),
  032307.

\bibitem{bennett1996ch}
Charles~H Bennett, David~P DiVincenzo, John~A Smolin, and William~K Wootters,
  \emph{Mixed state entanglement and quantum entanglement}, Physical Review A
  \textbf{54} (1996), 3824.

\bibitem{bkk1}
C{\'e}dric B{\'e}ny, Achim Kempf, and David~W Kribs, \emph{Generalization of
  quantum error correction via the {H}eisenberg picture}, Physical Review
  Letters \textbf{98} (2007), 100502.

\bibitem{bkk2}
\bysame, \emph{Quantum error correction of observables}, Physical Review A
  \textbf{76} (2007), 042303.

\bibitem{bkk3}
\bysame, \emph{Quantum error correction on infinite-dimensional {H}ilbert
  spaces}, Journal of Mathematical Physics \textbf{50} (2009), 062108.

\bibitem{boykin}
P.~Oscar Boykin and Vwani Roychowdhury, \emph{Optimal encryption of quantum
  bits}, Physical Review A \textbf{67} (2003), 042317.

\bibitem{braunstein2011zero}
Samuel~L Braunstein, David~W Kribs, and Manas~K Patra, \emph{Zero-error
  subspaces of quantum channels}, Information Theory Proceedings (ISIT), 2011
  IEEE International Symposium on, IEEE, 2011, pp.~104--108.

\bibitem{choi1}
Man~Duen Choi, \emph{A {S}chwarz inequality for positive linear maps on
  {$C^{\ast}$}-algebras}, Illinois Journal of Mathematics \textbf{18} (1974),
  565--574.

\bibitem{c-j-k}
Man-Duen Choi, Nathaniel Johnston, and David~W Kribs, \emph{The multiplicative
  domain in quantum error correction}, Journal of Physics A: Mathematical and
  Theoretical \textbf{42} (2009), 245303.

\bibitem{church}
Amber Church, David~W. Kribs, Rajesh Pereira, and Sarah Plosker, \emph{Private
  quantum channels, conditional expectations, and trace vectors}, Quantum
  Information \& Computation \textbf{11} (2011), 774--783.

\bibitem{cklt}
Jason Crann, David~W Kribs, Rupert~H Levene, and Ivan~G Todorov, \emph{Private
  algebras in quantum information and infinite-dimensional complementarity},
  Journal of Mathematical Physics \textbf{57} (2016), 015208.

\bibitem{devetak2005capacity}
Igor Devetak and Peter~W Shor, \emph{The capacity of a quantum channel for
  simultaneous transmission of classical and quantum information},
  Communications in Mathematical Physics \textbf{256} (2005), 287--303.

\bibitem{gottesman1996d}
Daniel Gottesman, \emph{Class of quantum error-correcting codes saturating the
  quantum {H}amming bound}, Phys. Rev. A \textbf{54} (1996), 1862.

\bibitem{grassl2017codes}
Markus Grassl, Sirui Lu, and Bei Zeng, \emph{Codes for simultaneous
  transmission of quantum and classical information}, Information Theory
  (ISIT), 2017 IEEE International Symposium on, IEEE, 2017, pp.~1718--1722.

\bibitem{holevo2}
Alexander~S Holevo, \emph{Complementary channels and the additivity problem},
  Theory of Probability \& Its Applications \textbf{51} (2007), 92--100.

\bibitem{holevo}
\bysame, \emph{Quantum systems, channels, information: A mathematical
  introduction}, Walter De Gruyter, 2013.

\bibitem{entng-brkng}
Michael Horodecki, Peter~W. Shor, and Mary~Beth Ruskai, \emph{Entanglement
  breaking channels}, Reviews in Mathematical Physics \textbf{15} (2003),
  629--641.

\bibitem{hsieh2010entanglement}
Min-Hsiu Hsieh and Mark~M Wilde, \emph{Entanglement-assisted communication of
  classical and quantum information}, IEEE Transactions on Information Theory
  \textbf{56} (2010), 4682--4704.

\bibitem{jochym1}
Tomas Jochym-O'Connor, David~W Kribs, Raymond Laflamme, and Sarah Plosker,
  \emph{Private quantum subsystems}, Physical Review Letters \textbf{111}
  (2013), 030502.

\bibitem{jochym}
\bysame, \emph{Quantum subsystems: exploring the complementarity of quantum
  privacy and error correction}, Physical Review A \textbf{90} (2014), 032305.

\bibitem{johnston}
Nathaniel Johnston and David~W Kribs, \emph{Generalized multiplicative domains
  and quantum error correction}, Proceedings of the American Mathematical
  Society \textbf{139} (2011), 627--639.

\bibitem{knill-lafl}
Emanuel Knill and Raymond Laflamme, \emph{Theory of quantum error-correcting
  codes}, Physical Review A \textbf{55} (1997), 900--911.

\bibitem{knill2000theory}
Emanuel Knill, Raymond Laflamme, and Lorenza Viola, \emph{Theory of quantum
  error correction for general noise}, Physical Review Letters \textbf{84}
  (2000), 2525.

\bibitem{kremsky2008classical}
Isaac Kremsky, Min-Hsiu Hsieh, and Todd~A Brun, \emph{Classical enhancement of
  quantum-error-correcting codes}, Physical Review A \textbf{78} (2008),
  012341.

\bibitem{kks}
Dennis Kretschmann, David~W Kribs, and Robert~W Spekkens, \emph{Complementarity
  of private and correctable subsystems in quantum cryptography and error
  correction}, Physical Review A \textbf{78} (2008), 032330.

\bibitem{kribsfixed}
David~W Kribs, \emph{Quantum channels, wavelets, dilations and representations
  of {$\mathcal{O}_n$}}, Proceedings of the Edinburgh Mathematical Society
  \textbf{46} (2003), 421--433.

\bibitem{klp}
David~W Kribs, Raymond Laflamme, and David Poulin, \emph{Unified and
  generalized approach to quantum error correction}, Physical Review Letters
  \textbf{94} (2005), 180501.

\bibitem{klpo}
David~W Kribs, Raymond Laflamme, David Poulin, and Maia Lesosky, \emph{Operator
  quantum error correction}, Quantum Information \& Computation \textbf{6}
  (2005), 383--399.

\bibitem{kribs-spekkens}
David~W. Kribs and Robert~W. Spekkens, \emph{Quantum error-correcting
  subsystems are unitarily recoverable subsystems}, Physical Review A
  \textbf{74} (2006), 042329.

\bibitem{kup}
Greg Kuperberg, \emph{The capacity of hybrid quantum memory}, IEEE Transactions
  on Information Theory \textbf{49} (2003), 1465--1473.

\bibitem{lkp}
J.~{Levick}, D.~W. {Kribs}, and R.~{Pereira}, \emph{{Quantum privacy and Schur
  product channels}}, Reports on Mathematical Physics \textbf{80} (2017),
  333--347.

\bibitem{ljklp}
Jeremy Levick, Tomas Jochym-O'Connor, David~W Kribs, Raymond Laflamme, and
  Rajesh Pereira, \emph{Private subsystems and quasiorthogonal operator
  algebras}, Journal of Physics A: Mathematical and Theoretical \textbf{49}
  (2016), 125302.

\bibitem{klappen2018}
Andrew Nemec and Andreas Klappenecker, \emph{Hybrid codes}, Information Theory
  Proceedings (ISIT), 2018 IEEE International Symposium on, IEEE, 2018,
  pp.~796--800.

\bibitem{nc00}
Michael~A Nielsen and Isaac~L Chuang, \emph{Quantum information \& quantum
  computation}, Cambridge University Press, 2000.

\bibitem{petz2007complementarity}
D{\'e}nes Petz, \emph{Complementarity in quantum systems}, Reports on
  Mathematical Physics \textbf{59} (2007), 209--224.

\bibitem{miza}
Mizanur Rahaman, \emph{Multiplicative properties of quantum channels}, Journal
  of Physics A: Mathematical and Theoretical \textbf{50} (2017), 345302.

\bibitem{shor1995pw}
Peter~W. Shor, \emph{Scheme for reducing decoherence in quantum computer
  memory}, Phys. Rev. A \textbf{52} (1995), R2493.

\bibitem{steane1996error}
Andrew~M Steane, \emph{Error correcting codes in quantum theory}, Physical
  Review Letters \textbf{77} (1996), 793.

\bibitem{stinespring}
W~Forrest Stinespring, \emph{Positive functions on {$C^{\ast} \ $}-algebras},
  Proceedings of the American Mathematical Society \textbf{6} (1955), 211--216.

\bibitem{Wein}
Mih{\' a}ly Weiner, \emph{On orthogonal systems of matrix algebras}, Linear
  Algebra and its Applications \textbf{433} (2010), 520 -- 533.

\bibitem{yard2005simultaneous}
Jon Yard, \emph{Simultaneous classical-quantum capacities of quantum multiple
  access channels}, PhD thesis, arXiv preprint quant-ph/0506050 (2005).

\end{thebibliography}
\bibliographystyle{amsplain}

\end{document}